\documentclass[11pt]{amsart}

\usepackage{amsthm,amssymb,amsfonts,latexsym,mathtools}
\usepackage{amsmath}
\usepackage{epsfig,color,hyperref,fancyhdr,geometry}
\usepackage{tikz-cd} 

\newtheorem{theorem}{Theorem}[section]
\newtheorem{lemma}[theorem]{Lemma}
\theoremstyle{definition}
\newtheorem{definition}[theorem]{Definition}
\newtheorem{proposition}[theorem]{Proposition}
\newtheorem{example}[theorem]{Example}
\newtheorem{remark}[theorem]{Remark}
\newtheorem{corollary}[theorem]{Corollary}

\usepackage{stmaryrd}
\usepackage{soul}

\def\kk{\Bbbk}

\def\R{\mathcal{R}}
\def\L{\mathcal{L}}

\def\X{\mathfrak{X}}
\def\Y{\mathfrak{Y}}
\def\trr{\triangleright}
\def\btr{\blacktriangleright}
\def\rtt{\rightthreetimes}
\def\trl{\triangleleft}
\def\btl{\blacktriangleleft}
\def\ltt{\leftthreetimes}
\def\id{\mbox{id}}
\def\nd{\mbox{End}}
\def\hm{\hbox{Hom}\,}
\def\alg{\hbox{Alg}\,}
\def\p{\partial}

\numberwithin{equation}{section}
\begin{document}
	 
\title{
	Quantum Calculi: differential forms and vector fields in noncommutative geometry
}
\author{Andrzej Borowiec\\
	{\tiny University of Wroclaw,  Institute of Theoretical Physics \\
		andrzej.borowiec@uwr.edu.pl}}


\maketitle
\large{\centerline{Devoted to the memory of Zbigniew Oziewicz,}
	\centerline{a good friend and an inspiring personality.\footnote{To appear in
	"Scientific Legacy of Professor Zbigniew Oziewicz: Selected Papers from the International Conference Applied Category Theory Graph-Operad-Logic” In Series on Knots and
	Everything (SKAE) in World Scientific Publishing.		
}}}

\begin{abstract}
	     In this paper, we revise the concept of noncommutative vector fields introduced previously in \cite{AB96,AB97}, extending the framework, adding new results and clarifying the old ones. Using appropriate algebraic tools certain shortcomings in the previous considerations are filled and made more precise. We focus on the correspondence between so-called Cartan pairs and first-order differentials.  The case of free bimodules admitting more friendly "coordinate description" and their braiding is considered in more detail. Bimodules of right/left universal vector fields are explicitly constructed.
\end{abstract}
 
  \section{Introduction} 
 Algebraic differential calculus is one of the cornerstones of noncommutative geometry present there from the very beginning in various contexts. In the original Connes approach it appeared as universal differential forms related to cyclic cohomology. \cite{Connes}. 
 Woronowicz \cite{ Woronowicz} introduced general bicovariant differential calculus on (matrix) quantum groups. Following his paper and \cite{Manin,Pusz,Wess},   many examples of covariant calculi on quantum plains have been elaborated  around this time.  These include our systematic study  of optimal calculi with partial derivatives on coordinate algebras \cite{bko94}, and their extension to higher orders \cite{bk95}.  
 
 In traditional differential geometry on smooth manifolds one has both, differential forms and vector fields, at our disposal.
 In a sense, vector fields should be considered as a primary object since they can be easily defined in purely algebraic ways as derivations of the algebra of smooth functions. Their action on function (via derivation) is, consequently, the main property of vector fields. On the other hand, they form a Lie module over the same algebra. Then differential one-forms can be introduced as a secondary objects (and then extended to entire Grassmann algebra), by taking the module dual. In fact, the derivations are module duals of so-called K\"{a}hler differentials, to which usual (smooth) differential one-forms are bidual.
 Unfortunately, generalization of this scheme to noncommutative setting has no unique answer, see, e.g., \cite{Michor96,Lunts97,GMS05,GS09}. The main reason is that derivations of noncommutative algebra do not form a module, so cannot be multiplied by elements of the algebra. But yet, they are closed under the commutator bracket. Consequently, they are equipped in a Lie algebra structure. This property  is preferred for some applications, see, e.g. \cite{Michor96,Manoh17,Marmo18,Kauffman21}.
 In contrast, in noncommutative geometry, differential forms are more fundamental and easier to handle. However, the notion of vector fields can also be introduced \cite{AB96,AB97} and turns out to be useful for some applications \cite{Hajac,Llena03,Beggs03,Pena08,Beggs14,Beggs20,Aschieri20}, including the recent studies of quantum Riemannian geometry \cite{BeggsMajid}.
 
 Our aim here is to enforce, improve and analyze in more adequate mathematical language the main arguments in favor of the definition introduced in \cite{AB96}. Our approach was motivated twofold. Firstly, by the idea of extending Cartan type duality between bimodule of one forms and vector fields.  Second, by the possibility of introducing the partial derivatives in a noncommutative setting and the directional covariant derivatives \cite{AB97}.
 
 A similar concept to employ module duality appeared in some papers around the same time, see, e.g. \cite{Peter96,Sitarz96,Dubois96,Aschieri96}.
 Our novelty was to provide an abstract algebraic definition which takes into account an action of vector fields (called Cartan pairs) on "functions" represented by elements of the algebra. This action assumes some generalized right or left handed version of the Leibniz rule combined with bimodule multiplications.  
 Then, it turns out that, by applying a module duality one can assign to any first-order calculus, a well-functioning notion (noncommutative generalization) of vector fields with a good operational meaning \cite{AB97,BeggsMajid}.
 
 But the situation is far to be symmetric. Unlike in standard differential geometry, the converse route from vector fields to the initial differential one-forms is not fully satisfying and needs much more technical assumptions, concerning their bimodule structure and/or special properties for a "ring of functions". The reasons for this asymmetry are rooted in the functorial properties of the module duality. 
 \subsubsection{Zbigniew Oziewicz (1941-2020)}
 Zbigniew Oziewicz, Professor of Wroclaw University, Institute of Theoretical Physics and Universidad Nacional Autonoma de Mexico, Facultad de Estudios Superiores Cuautitlan,
 was a very extraordinary man with a great passion to science.
 In particular, Zbyszek was a big fan of advanced mathematical methods believing that developing new tools will allow to solve almost all problems encountered in theoretical physics. His {\it idee fixe} was to use abstract mathematical techniques in order to reformulate physical problems in a coordinate-free and basis independent manner. For these reasons, he became interested in noncommutative geometry at the beginning 1990s when our cooperation started.  In fact, I  was inspired with this subject by Zbyszek and his enthusiasm. He invited Slava Kharchenko, that time an algebraist from Novosibirsk,  to participate in the program. Several joint articles, influenced by \cite{Woronowicz,Manin,Pusz,Wess} and developing general formalism of differential calculi on coordinate algebras with values in free bimodules, were published 
 \cite{bko94,bk95,bmo92,bk95a,bk95b,bko97}. In 1992, Zbyszek organized an international conference devoted to various aspects of mathematical physics. He invited many leading and active till now experts of noncommutative geometry \cite{OJB93}. Soon after, he left Poland for Mexico where he spent the rest of his academic career.  Working at UNAM, Zbigniew organized, inter alia, a series of interdisciplinary and internationally recognized conferences/workshops under the common motto \textit{ Graph-Operad-Logic}, even if it suggests a rather specialized topic.
 \section{Preliminaries and notations}
 Throughout this paper, we shall use a standard algebraic notation borrowed from e.g., \cite{Bourbaki}, \cite{Kasch}. 
 We work with  vector spaces (mostly infinite-dimensional) over fixed algebraically closed field $\kk$ (usually $\mathbb{C}$), all maps between two spaces $V$ and $U$ are $\kk$-linear maps i.e., belong to $\hm(V,U)\equiv \hm_\kk(V,U)$. The tensor product $\otimes\equiv \otimes_{\kk}$ if not otherwise stated.
 Denoting by $V^\bullet=\hm(V,\kk)$ the dual space one comes to the following canonical injective homomorphism (that becomes an isomorphism only in the case $\dim U<\infty$.)
 \begin{equation}\label{i1}
 	\hm(V,U)\hookrightarrow \hm(U^\bullet, V^\bullet), \quad \Phi \mapsto \Phi^\bullet\,,
 \end{equation}
 given by the transpose map $\Phi^\bullet(\alpha)(v)=\alpha(\Phi(v))$, where $\alpha\in U^\bullet, v\in V$. 
 The injective homomorphisms  $V\hookrightarrow V^{\bullet\bullet}$ is bijective only for $\dim V<\infty$. 
 Other canonical injective homomorphisms, which become isomorphisms when at least one of the ingredients is finite-dimensional \footnote{$V^\bullet\otimes V$ does not contain the  identity $\id_V$ for $\dim V=\infty$.},  has the form
 \begin{equation}\label{i2}
 	V^\bullet\otimes U\rightarrow \hm(V, U)\,. 
 \end{equation}
 In fact, $ V^\bullet\otimes U\cong \hm^{{\rm fin}}(V, U)\doteq \{ \Phi \in \hm(V, U)| \dim {\rm Im}(\Phi)<\infty\}$.
 We recall that in the (strictly monoidal) category of vector spaces each short exact sequence
 \begin{equation}\label{e1}
 	0\longrightarrow V'\stackrel{f}{\longrightarrow} V\stackrel{g}{\longrightarrow} V''\longrightarrow 0
 \end{equation}
 splits i.e., $V\cong V'\oplus V''$ and $V''\cong V/V'$. It means that $f$ is an injective map (monomorphism), $g$ is a surjective map (epimorphism) and
 ${\rm Im}f=\ker g\equiv V'$. Furthermore,  there exist the splitting maps
 $ V\stackrel{f_0}{\longrightarrow} V'$ and $ V''\stackrel{g_0}{\longrightarrow} V$ such that $f_0\circ f=\id_{V'}$ and $g\circ g_0=\id_{V''}$.
 These, of course, imply exactness (and splitting) for the reversed sequence:
 \begin{equation}\label{e2}
 	0\longrightarrow V''\stackrel{f_0}{\longrightarrow} V\stackrel{g_0}{\longrightarrow} V'\longrightarrow 0\,.
 \end{equation}
 Moreover, for any vector space $U$ the following sequences are exact and split along with \eqref{e1}:
 \begin{subequations}
 	\begin{align}
 		0\longrightarrow \hm(U,V') \stackrel{\bar f}{\longrightarrow} \hm(U,V)\stackrel{\bar g}{\longrightarrow} \hm(U,V'')\longrightarrow 0\,,\label{e3a}\\
 		0\longrightarrow \hm(V'',U)\stackrel{g^*}{\longrightarrow} \hm(V,U)\stackrel{ f^*}{\longrightarrow} \hm(V',U)\longrightarrow 0\,,\label{e3b}\\
 		0\longrightarrow V'\otimes U\quad\stackrel{f\otimes \id_U}{\longrightarrow} V\otimes U\stackrel{g\otimes\id_U}{\longrightarrow}\quad V''\otimes U\longrightarrow 0\,,\label{e3c}
 	\end{align}
 \end{subequations}
 where $\bar f(\phi)=f\circ\phi$, $\bar g(\phi)=g\circ\phi$, $ g^*(\psi)=\psi\circ g$, $f^*(\psi)=\psi\circ f$. Replacing $U$ by $\kk$ in the second row one obtains the exact sequence for the dual spaces with pullback maps.
 
 Fix some unital associative (possibly noncommutative, infinite-dimensional) linear algebra $R$ over $\kk$ with a (surjective) multiplication map denoted by $\mu: R\otimes R\rightarrow R$. In order to simplify the notation we shall write $1\equiv 1_R\in R$. We are interested in unital left/right/bi modules over $R$. So, a module means unital $R$-module and at the same time  a $\kk$-vector space.  Sometimes, for clarity, we use the notation $M_R$ (respectively, $\ _R\! M$) to denote a right (respectively left)  $R$-module. 
 Free right (respectively, left) modules are of the form $V\otimes R$ (respectively, $R\otimes V$) for some vector space $V$ of generators.
 The space of $R$-linear maps between two right modules is denoted as $\hm_R(M_R, N_R)\subset \hm(M, N)$.
 There are very useful canonical isomorphisms of vector spaces
 \begin{equation}\label{ii3}
 	\hm_R(R\otimes M_R,N_R)\cong \hm(R,\hm_R(M_R,N_R))\,,
 \end{equation}
 and for left modules
 \begin{equation}\label{ii4}
 	\hm_R(_R\! M\otimes R,\ _R\!N)\cong \hm(R,\hm_R(_R\! M,\ _R\!N))\,.
 \end{equation}
 Some other functorial properties presented above have to be suitably modified when applied in the category of modules over the ring $R$.\footnote{Some modifications may depend on the particular properties both modules involved as well as the rings. Usually they are preserved in the case of free finitely generated modules or when the ring has special properties \cite{Bourbaki,Kasch}.} 
 
 Next, we recall that for right module $M=M_R$ its module (right) dual (over the ring $R$) is defined by $(M_R)^\trr=\hm_R (M_R, R_R)$ and bears a canonical left module structure $(r\phi)(m.r')=r\phi(m.r')=r\phi(m)r'$. In terms of a canonical bilinear  pairing  $<|>:M^\trr\otimes M\rightarrow R$, provided by the evaluation map, one obtains
 \begin{equation}\label{i3a}
 	<r.\phi| m.r'>=r.\phi(m).r' \ \,. \end{equation}
 For example $(R_R)^\trr =\,_R\! R$ \footnote{The ring itself can be considered as the right $R_R$, left $_R\! R$ or bimodule $_R\! R_R$.}. 
 By the same token one can define a left dual $(_R\! N)^\trl=\hm_R (_R\! N, _R\! R)$ with the canonical right module structure $<r'.m| \xi.r>=r'.\xi(m).r$. Thus double dual (bidual) of right (respectively, left) module is again a right (res. left module). One has a canonical map
 \begin{equation}\label{i3b}
 	M_R\longrightarrow (M_R^\trr) ^\trl\,,\quad m\mapsto m^{\trr\trl} \,,
 \end{equation}
 where $ m^{\trr\trl}(\phi)=\phi(m)$ for each $\phi\in M_R^\trr$.  It is injective for torsionless modules.    $M_R$ is called reflexive if $M_R\cong (M_R\,^\trr) ^\trl$ (resp. $_R\! N\cong (_R\! N^\trl)^\trr$).
 
 There is also an analogue of transpose homomorpism which is a linear map between two vector spaces
 (if $M, N$ are right modules then $M^\trr, N^\trr$ are left ones and vice versa)
 \begin{equation}\label{i4}
 	\hm_R(M,N)\rightarrow \hm_R(N^\trr, M^\triangleright), \quad \Psi \mapsto \Psi^\trr\,,
 \end{equation} 
 given by the transposition map $\Psi^\trr (a)(m)\doteq a(\Psi(m))$, or 
 \[<\Psi^\trr(a)| m>=<a|\Psi(m)>\,,\] where $a\in N^\trr, m\in M $. 
 It has many interesting properties, e.g. if $\Psi$ is surjective then $\Psi^\trr$ is injective (but not vice versa in general), cf. \cite{Bourbaki,Kasch}.  
 \begin{proposition}\label{key1} (see, e.g. \cite{Bourbaki,Kasch})
 	The exact sequence of right modules:
 	\begin{equation}\label{i7}
 		K\stackrel{f}{\longrightarrow} M\stackrel{g}{\longrightarrow} N
 		\longrightarrow 0
 	\end{equation}
 	implies that the following "dual'' sequence of left modules (cf. \eqref{e3b})
 	\begin{equation}\label{i8}
 		0\longrightarrow 
 		N^\trr\stackrel{g^\trr}{\longrightarrow} M^\trr\stackrel{f^\trr}{\longrightarrow} K^\trr
 	\end{equation}
 	is also exact. Here, in contrast to \eqref{e3b}, $f$ is not necessarily injective and $f^\trr$ is not necessarily surjective. This conclusion does not change, in general, even when \eqref{i7} extends to the short exact sequence $0\rightarrow K\rightarrow M\rightarrow N\rightarrow 0$ i.e., when $f$ is injective. However, if $f$ is injective and \eqref{i7}  splits then \eqref{i8} extends to a short exact sequence and splits as well. Of course, the same conclusion holds if one replaces the right by the left and vice versa.
 \end{proposition}
 
 An interesting question is how to dualize bimodules. It appears that right (respectively, left) dual of a bimodule $M^\trr=\hm_{(-,R)}(M, R_R)$ (respectively,  $M^\trl=\hm_{(R,-)}(M, _R\!R)$) is again a bimodule with the right (respectively, left) module structure given by 
 \begin{equation}\label{i9}
 	<\phi.r| m>=\phi(r.m),\quad (\mbox{respectively,} <m|r.\phi>=\phi (m.r))\,.
 \end{equation}
 The left (respectively, right) module structure remains the same as, e.g. in \eqref{i3a}. Moreover, if $\Psi\in \hm_{(R,R)}(M, N)$ is a bimodule map (cf. \eqref{i4} ), its right (respectively, left) transpose $\Psi^\trr\in \hm_{(R,R)}(N^\trr, M^\trr)$ ($\Psi^\trl\in \hm_{(R,R)}(N^\trl, M^\trl)$) is a bimodule map as well.  
 Further, one needs the following \footnote{More general setting has been proposed in \cite{bko97}.}:
 \begin{proposition}\label{lmod}
 	Let $M=M_R$ be a right module and let $\beta\in \hm_R(R\otimes M, M)$ be a right modules map. Then the following statements are equivalent:
 	\begin{enumerate}
 		\item[(i)] $\beta$ determines a left module structure  making $M$ a bimodule i.e., $(r.m).r'\doteq \beta(r\otimes m).r'\equiv \beta(r\otimes m.r')\doteq r.(m.r')$
 		\item[(ii)] $\beta(1\otimes m)=m$ and $\beta(rr'\otimes m)=\beta(r\otimes\beta(r'\otimes m))$, or equivalently, $\beta\circ (\mu\otimes \id_M)=\beta\circ (\id_R\otimes \beta)$ on $R\otimes R\otimes M$.
 		\item[(iii)] there is a $\kk$-algebra homomorphism $\hat\beta\in\alg(R,\nd_R(M))$ i.e., $\hat\beta(r_1r_2)=\hat\beta(r_1)\circ\hat\beta(r_2)$, where $\hat\beta(r)(m)\doteq \beta(r\otimes m)$ and $\hat\beta(1)=\id_M$.
 	\end{enumerate}
 	In addition, if $M=V\otimes R$ is a free right module generated by the vector space $V$ then:
 	\begin{enumerate}
 		\item[(iv)] there is a braiding map $\tilde\beta \in\hm(R\otimes V, V\otimes R)$ such that  $\beta(r\otimes v\otimes r')\doteq\tilde\beta(r\otimes v).r'$ satisfies  (ii). In particular $\tilde\beta(1\otimes v)=v\otimes 1$.	\end{enumerate} 
 \end{proposition}
 \begin{proof}
 	(i)$\Leftrightarrow$(ii) is obvious and follows directly from (unital) bimodule axioms.
 	To see (iii) one has to take into account the following canonical isomorphisms
 	\[ \hm_R(R\otimes M_R, M_R)\cong \hm(R, \nd_R(M_R)) \supset\alg(R,\nd_R(M_R))\,.\]
 	Similarly, (iv) follows from $\tilde\beta(r\otimes v)=\beta(r\otimes v\otimes 1_R)=\hat\beta(r)(v\otimes 1_R)$
 	\begin{eqnarray}{}
 		\hm_R(R\otimes V\otimes R, V\otimes R)&\cong &\hm(R\otimes V,V\otimes R)\,,
 	\end{eqnarray}
 	since  $\hm_R(V\otimes R, M_R)\cong \hm(V, M_R)$. 
 \end{proof}
 \begin{corollary}\label{coord1} (Coordinate description, cf. \cite{bko94,bk95,bk95b}.) Let $V$ be a vector space with a basis $\{\xi^i\}_{i\in I}$ ($I$ being a finite or infinite set of indices). Let $M_R=V\otimes R$ be a right free module. Thus, by using Einstein summation convention, any element $\omega=\xi^i\otimes \omega_i$ can be uniquely represented by its (covariant) coordinates $\{\omega_i\in R, i\in I\}$  (the cardinality of $I$ is the dimensionality of $V$) \footnote{All sums have to be finite, therefore only finite number of coordinates is allowed to be nonzero.}. With this notation $\beta(r\otimes \omega)=\xi^k\otimes \hat\beta_k^i(r)\omega_i$, where
 	\begin{equation}\label{coord2}
 		r.\omega.s=\xi^k\otimes \hat\beta_k^i(r)\omega_i\,s\,,\quad \hat\beta_k^i(rr')=\hat\beta_k^l(r)\hat\beta_l^i(r')\,, \quad \tilde\beta (r\otimes \xi^i)=\xi^k\otimes\hat\beta_k^i(r)
 	\end{equation}	
 	and the matrix $\hat\beta^k_i(r)$ with entries in $R$ represents the element $$\hat\beta(r)\in \nd_R(M_R)\cong \hm(V,V\otimes R)\,,$$ $\hat\beta(r)(\xi^i)=\xi^k\otimes \hat\beta_k^i(r)$. Obviously, $\hat\beta^i_k\in\nd(R)$ (each column contains only the finite number of nonzero elements) 
 	and  $\hat\beta^i_k(1)=1\,\delta^i_k$. In some applications, elements $\xi^i=dx^i$ can be related with generators $x^i$ of the algebra $R$ by the calculus $d:R\rightarrow V\otimes R$, see the next sections.  
 \end{corollary} 
 \begin{corollary}\label{lmod}
 	Similarly, for left module $_R\! N$ and $\alpha\in \hm_R(N\otimes R,N)$ one gets right (bimodule) structure $\iff$ $\alpha(n\otimes 1)=n$ and $\alpha(n\otimes rr')=\alpha(\alpha(n\otimes r)\otimes r')$, or equivalently $\alpha\circ (\id_N\otimes \mu)=\alpha\circ (\alpha\otimes \id_R)$ on $N\otimes R\otimes R$. In this case $\hat\alpha\in\alg(R^{op},\nd_R(N))$. Considering $(R\otimes V)_\alpha$ as a bimodule, the braiding
 	$\tilde\alpha \in\hm(V\otimes R, R\otimes V) $ satisfies $\tilde\alpha (v\otimes r)=\alpha(1\otimes v\otimes r)=\hat\alpha(r)(1\otimes v)$ and $\tilde\alpha(v\otimes 1)=1\otimes v$.   
 \end{corollary}
 \begin{corollary}\label{dual}
 	Let $N=M^\trr$ and $M=\,_\beta\! M_R$  be a bimodule with left module structure as in Proposition \ref{lmod} (ii). Then $N=\,_{R\!}N_\alpha$ has right module structure given by $\hat\alpha(r)=\hat\beta(r)^\trr$ i.e.,  $_R\! N_\alpha=(_\beta\! M_R)^\trr$. 
 \end{corollary}		
 Consider the case of free right module $M_R=V\otimes R$ with more details. Its (right) dual can be identify as
 \begin{equation}\label{i5}
 	(V\otimes R)^\trr =\hm_R(V\otimes R, R_R)\cong \hm(V, R)\,
 \end{equation}
 which has a canonical left module structure: $(r.\phi)(v\otimes r')=r\phi(v\otimes 1_R)r'$, where
 $(r.\tilde\phi)(v)\doteq r\phi(v\otimes 1_R)$, $r.\tilde\phi : V\rightarrow R$. In fact, one has more general canonical isomorphism \begin{equation}\label{ii5}
 	\hm_R(V\otimes R, M_R)\cong \hm(V, M)
 \end{equation}
 since $\hm_R(R, M_R)\cong  M_R$, cf. \eqref{ii3}. 
 It means that dual of free module is not necessary free, unless some special conditions are satisfied. Instead we always have injective homomorpism of right (respectively, left \footnote{For a free left module $(R\otimes V)^\trl =\hm_R(R\otimes V,\ _R\! R)\cong \hm(V, R)$ with the right module structure $(\phi.r)(r'\otimes v)=r'\phi(1\otimes v)r$, so
 	$(\tilde\phi.r)(v)\doteq \phi(1\otimes v)r$.}) modules
 \begin{equation}\label{iii5}
 	V^\bullet\otimes R\hookrightarrow \hm(V, R)\end{equation} 
 (respectively, $R\otimes V^\bullet \hookrightarrow \hm(V, R)$).
 If at least one of the spaces ($V$ or $R$) is finite-dimensional then 
 $\hm(V, M)\cong V^\bullet\otimes M\cong M\otimes V^\bullet$. In such a case, $\hm(V, R)\cong V^\bullet\otimes R\cong R\otimes V^\bullet$ i.e., one has free dual module.  Thus free modules are reflexive in the case $\dim V<\infty$.
 \begin{corollary}\label{bdual}
 	Since $R\otimes V\cong V\otimes R$ as vector spaces they can be also isomorphic as bimodules $(R\otimes V)_\alpha\cong\ _\beta\! (V\otimes R)$.
 	For this to happen it is sufficient and enough that there is a braiding $\mathfrak{L}\in 
 	{\rm Iso}(V\otimes R, R\otimes V)$ such that the following two conditions are satisfied on $V\otimes R\otimes R$
 	\newline(1) $\mathfrak{L}\circ (\id_V \otimes\mu)=
 	\alpha\circ(\mathfrak{L}\otimes \id_R)$   and
 	\newline(2) $\mathfrak{L}^{-1}\circ (\mu \otimes\id_V)=
 	\beta\circ(\id_R \otimes \mathfrak{L}^{-1} )$ \\
 	One can observe that, particularly, the first condition is satisfied by $\mathfrak{L}=\tilde\alpha$	while the second by $\mathfrak{L}^{-1}=\tilde\beta$. Thus one obtains
 \end{corollary}	
 \begin{lemma}\label{L1}
 	If $\tilde\beta =\tilde\alpha^{-1}$ then  the bimodule isomorphism $(R\otimes V)_\alpha\cong\ _\beta\! (V\otimes R)$ is provided by the braiding $\tilde\alpha$.
 \end{lemma}	
 \section{Quantum calculi in terms of differential one-forms}
 Let $\bar R\doteq R/\kk$ be a quotient vector space and denote by $D:R\rightarrow\bar R$ the canonical surjection $r\mapsto\bar r$, where $\bar r$ denotes the corresponding equivalence class ($r\sim r'$ $\Leftrightarrow$ $r-r'\in 1 \kk$) and $1\kk=\ker D_0$. Consequently, $\bar 1=0$. It is related with the short exact sequence of vector spaces
 \begin{equation}\label{f1}
 	0\longrightarrow \kk\stackrel{1}{\longrightarrow} R\stackrel{D}{\longrightarrow} \bar R\longrightarrow 0
 \end{equation}
 which splits: $R\cong \bar R\oplus \kk$ as a vector space; $ \mbox{codim}\bar R=1$.  
 \begin{remark}\label{}
 	This splitting is, of course, up to some isomorphism which depends upon a splitting injective map $D_0: \bar R \rightarrow R$ (some choice of representative in each equivalence class) such that $D\circ D_0=\id_{\bar R}$.
 	In some special case, not assumed here, $\bar R$ might be a nonunital subalgebra  of $R$.
 \end{remark}	
 Applying the functor $\otimes R$ one comes to new exact sequence
 \begin{equation}\label{f2}
 	0\longrightarrow R\stackrel{1\otimes\id_R}{\longrightarrow} R\otimes R\stackrel{D\otimes \id_R}{\longrightarrow} \bar R\otimes R\longrightarrow 0
 \end{equation} 
 which splits  $R\otimes R\cong R\oplus (\bar R\otimes R)$. It is not difficult to see that \eqref{f2} is at the same time an exact sequence of right modules and one obtains the splitting of right modules. The splitting map for $1\otimes\id_R$ is the multiplication
 map $R\otimes R\stackrel{\mu}{\rightarrow}R$.  An injective map $\mathfrak{I}: \bar R\otimes R\rightarrow R\otimes R$ defined as $\mathfrak{I}(\bar r\otimes s)\doteq r\otimes s - 1\otimes rs$ \footnote{Here, the right hand side does not depend on the choice of representative $r\in\bar r$ i.e., it is invariant with respect to $r\mapsto r+\kk$. This has to be always the case when one deals with the variable $\bar r$.} does satisfy the property $(D\otimes\id_R)\circ \mathfrak{I} =\id_{\bar R\otimes R}$. 
 The corresponding splitting of right modules takes the form
 \begin{equation}\label{f2a}
 	R\otimes R \equiv {\rm Im}\mathfrak{I} \oplus R\,,
 \end{equation}
 i.e. $r\otimes s=(r\otimes s - 1\otimes rs)+1\otimes rs$.
 Therefore, the last sequence of right modules can be reversed to provide an exact sequence of right modules (${\rm Im}\mathfrak{I}=\ker\mu$)
 \begin{equation}\label{f3}
 	0\longrightarrow \bar R\otimes R\stackrel{\mathfrak{I}}{\longrightarrow}R\otimes R\stackrel{\mu}{\longrightarrow}
 	R\longrightarrow 0\,
 \end{equation} 
 which is usually a starting point for considering a universal differential calculus and its relation with the graded universal differential algebra, Hochschild homology, cyclic cohomology, etc., see, e.g. \cite{Connes,Karoubi,Cuntz,Quillen}.
 \begin{remark}\label{rr1}
 	In fact, ${\rm Im}\mathfrak{I}=\ker\mu$ is a subbimodule of $R\otimes R$. Thus $_\chi\! (\bar R\otimes R)\stackrel{\mathfrak{I}}{\longrightarrow}\ker\mu$ is a bimodule isomorphism if we 
 	set the canonical left braiding  (cf. Corollary \ref{bdual}) map $\tilde\chi (r\otimes\bar s)\equiv r.(\bar s\otimes 1)\doteq \overline{rs}\otimes 1- \bar r\otimes s$ on $ R\otimes\bar  R$. However, \eqref{f3}  is not a left module  splitting. The reason is that the first arrow in \eqref{f2} is not a left module map. 
 \end{remark}
 \begin{remark}\label{rr2}
 	By the same token,  applying $R\otimes $ to \eqref{f1} one obtains new exact sequence of left modules
 	\begin{equation}\label{f2b}
 		0\longrightarrow R\stackrel{\id_R\otimes 1}{\longrightarrow} R\otimes R\stackrel{\id_R\otimes D}{\longrightarrow}  R\otimes \bar R\longrightarrow 0\,,
 	\end{equation} 
 	which splits and  generates  
 	\begin{equation}\label{f4}
 		0\longrightarrow  R\otimes \bar R\stackrel{\mathfrak{T} }{\longrightarrow}R\otimes R\stackrel{\mu}{\longrightarrow}
 		R\longrightarrow 0\,
 	\end{equation}
 	which splits as well, where $\mathfrak{T}( r\otimes \bar t)\doteq rt\otimes 1-r\otimes t$. Again, defining  a right module structure (right canonical braiding)
 	$\tilde\gamma (\bar r\otimes s)\equiv (1\otimes \bar r).s\doteq 1\otimes\overline{rs}- r\otimes \bar s$ on $ R\otimes \bar R$ we make it a subbimodule   ${\rm Im}\mathfrak{T}\cong \ker\mu\subset R\otimes R$. However, by similar reasons, the splittings \eqref{f2b}, \eqref{f4} are left module splittings.  
 \end{remark} 
 \begin{remark}\label{dR}
 	The left $\tilde\chi$ and right $\tilde\gamma$ canonical braidings are mutually inverse to each other. Therefore, according to the Lemma  \eqref{L1} bimodules $(R\otimes \bar R)_\gamma\cong\ _\chi\! (\bar R\otimes R)$ are isomorphic.
 \end{remark}	
 Note that $\ker\mu=\{\sum_{i=1}^{n} r_i\otimes r'_i\in R\otimes R| \sum_{i=1}^{n} r_i r'_i=0, n\in\mathbb{N} \}$ is a subbimodule of the bifree bimodule $R\otimes R$ one obtains the exact sequence of bimodules
 \begin{equation}\label{f5}
 	0\longrightarrow  \ker \mu \stackrel{\mathfrak{C}}{\longrightarrow} R\otimes R\stackrel{\mu}{\longrightarrow}
 	R\longrightarrow 0\,,
 \end{equation}
 where $\mathfrak{C}$ is the inclusion map. 
 \begin{proposition}\label{}
 	This sequence admits left (respectively, right) module splitting depending upon a choice of $\mu_0$. 
 	There is no bimodule splitting of \eqref{f5}.
 \end{proposition} 
 \begin{proof} 
 	The sequence \eqref{f5} splits since $R\otimes R= \ker\mu \oplus R$. In order to obtain right/left/bi-module splitting
 	one needs that splitting maps $\mu_0: R\rightarrow R\otimes R$ and $\mathfrak{C}_0: R\otimes R\rightarrow \ker\mu$ satisfying $\mu_0\circ \mu=\id_R$, $\mathfrak{C}\circ\mathfrak{C}_0=\id_{\ker\mu}$ share the same properties.
 	If $\mu_0(r)=r\mu_0(1)=r\otimes 1$ is left module map then it cannot be at the same time a right module (bimodule) map and
 	$r\otimes 1\notin \ker\mu$.  Consequently, 
 	$\mathfrak{C}\circ\mathfrak{C}_0(r\otimes r') =\mathfrak{C}\circ\mathfrak{C}_0(r\otimes r'-rr'\otimes 1+ rr'\otimes 1')=r\otimes r'-rr'\otimes 1\in\ker\mu$. Since ${\rm Im}\mathfrak{C}\circ\mathfrak{C}_0$ is a left submodule then $\mathfrak{C}_0$ has to be left module map. \end{proof}
 We are now in position to recall the well-known definition.
 \begin{definition}\label{d1} (see, e.g. Woronowicz \cite{Woronowicz,Bourbaki})
 	A \emph{First Order Differential Calculus (FOC in short)} over $R$ is a
 	pair $(d,\Omega)$ where:
 	\newline(1)
 	$\Omega$ is an $R$--bimodule,
 	\newline(2)
 	$d\colon{R}\longrightarrow\Omega$ is a linear map called differential satisfying the Leibniz rule i.e., $d(rs)=d(r).s+r.d(s)$ for any $r,s,\in{R}$, and
 	\newline(3)
 	$\Omega$ is generated, as $R$--bimodule, by elements $\{d(r)\colon\;r\in{R}\}$.\\  Elements of $\Omega$ are called \emph{1-forms}.
 \end{definition}
 \begin{remark}
 	$\Omega$ is also generated, as left or
 	right $R$--module, by the elements $\{d(a)\colon\;a\in{R}\}$ since for any $a$, $b\in{R}$ we
 	have the identity $d(a)b=d(ab)-ad(b)$. Clearly, $\ker d\supset 1\kk$. If $\ker d = 1\kk$
 	then the corresponding FOC is called connected \cite{BeggsMajid}.
 	It is also known that FOC extends to a graded differential algebra with the property $d^2=0$, which is not discussed here (see, e.g. \cite{Bourbaki,bk95a,Cuntz}).
 \end{remark}
 We recall that a morphism between two FOC $d_1:R\rightarrow\Omega_1$ and  $d_2:R\rightarrow \Omega_2$ is a bimodule map $\mathfrak{L} \in \hm_{(R,R)}(\Omega_1,\  \Omega_2)$ such that $d_2=\mathfrak{L}\circ d_1\,$.
 \begin{proposition}\label{uFOC} The following first order calculi are mutually isomorphic
 	\newline(1) $ R\stackrel{D\otimes 1}{\longrightarrow}\ _\chi\! (\bar R\otimes R)$\,,
 	\newline(2) $ R\stackrel{1\otimes D}{\longrightarrow} (R\otimes \bar R)_\gamma$\,,
 	\newline(3)  $ R\stackrel{d_\mu}{\longrightarrow} \ker\mu$, where $d_\mu r=r\otimes 1-1\otimes r$\,.
 	\begin{proof}   
 		We already know that  the bimodules  (1) and (2) are isomorphic by the canonical braiding 
 		$\tilde\gamma (\bar r\otimes s)=  1\otimes\overline{rs}-r\otimes \bar s$ with the inverse
 		$\tilde\chi (r\otimes \bar s)= \overline{rs}\otimes 1-\bar r\otimes s$.  Since $\bar r\equiv Dr$ then   $\tilde\gamma (D r\otimes s)= 1\otimes Dr$. Similarly, $\mathfrak{T}\circ\tilde\gamma$ provides an isomorphism between (1) and (3). 
 	\end{proof}
 \end{proposition} 
 \begin{example}\label{}
 	(Inner differential) $\Omega$ is any $R$--bimodule, and $\omega\in \Omega$ any element.
 	Thus $d_\omega (a):=a.\omega-\omega.a\equiv [a\,,\omega]$ defines FOC.
 	In this sense the differential $d_\mu r= [r, 1\otimes 1]$ is "outer" since $1\otimes 1\notin \ker(\mu)$.\end{example}
 A fundamental meaning of FOC from the Proposition \eqref{uFOC} is provided by the following
 universality  property:
 \begin{proposition}\label{univ} (see, e.g. \cite{Woronowicz,Bourbaki})\ 
 	The following statements are equivalent:
 	\begin{enumerate}
 		\item[(a)]  $d:R\rightarrow \Omega$ is FOC,
 		\item[(b)] There is a surjective bimodule map (bimodule epimorphism) $\pi: \ker\mu \rightarrow \Omega$ such that $dr=\pi (d_\mu r)$, 
 		\item[(b)] There exist a subbimodule $\bar\Omega\subset\ker\mu$ such that the following short exact sequence of bimodules is exact
 		\begin{equation}\label{univ1}
 			0\longrightarrow \bar\Omega\longrightarrow\ker\mu  \stackrel{\pi}{\longrightarrow}  \Omega\longrightarrow 0\,.
 		\end{equation}
 	\end{enumerate}	 
 \end{proposition}
 \begin{proof} We follow the scheme: (a)$\Rightarrow$(b)$\Rightarrow$ (c)$\Rightarrow$(a).\\
 	(a)$\Rightarrow$(b)		If $\sum_{i}r_i\,s_i=0$ then $\sum_{i}r_i\otimes{s_i}=-\sum_{i}d_\mu\,(r_i).{s_i}=\sum_{i}r_i.
 	d_\mu s_i\,$, where 
 	$d_\mu r_i= r_i\otimes 1 - 1\otimes r_i$, and one defines
 	$\pi(\sum_{i}r_i\otimes{s_i})\doteq	-\sum_{i}d\,(r_i).{s_i}=\sum_{i} r_i.d(s_i)$. It turns out to be a bimodule map. Let now $\omega=\sum_{i}d_\mu\,(r_i).{s_i}$ be any element of $\Omega$. Then $\tilde\omega=\sum_{i}(r_i\otimes{s_i}-1\otimes r_is_i)\in\ker\mu$ and $\pi(\tilde\omega)=\omega$. \\
 	(b)$\Rightarrow$(c) Define $\bar\Omega\doteq\ker\pi$ then $\Omega\cong \ker\mu/\ker\pi$ and one has the short exact sequence obvious.\\
 	(c)$\Rightarrow$(a) Finally, defining $d\doteq\pi\circ d_\mu$ we close the proof.
 \end{proof}  
 \begin{remark}\label{qFOC}  
 	This  sequence does not split, in general, neither by the left nor by the right modules. Of course, some splitting linear maps $\pi_0: \omega \rightarrow \ker\mu$ exist  but not necessary as module maps.		
 \end{remark}
 \begin{example}\label{standard}	 The standard exterior Cartan calculs on a smooth manifold $\mathcal{M}$ with the commutative algebra $R=C^\infty(\mathcal{M})$ can not be obtained in this way. Instead, it is obtained by the quotient of a universal K\"ahler differential $d_K:R\rightarrow \Omega^1_K(R)$, which is universal in a class of calculi with the left and right module structures the same.  The universal K\"ahler differential always exists for commutative rings, see, e.g. \cite{Kunz}.\end{example}
 We would like to mention that modules over commutative algebras are generalized by the so-called central bimodules \cite{Michor96,Conti22}. Bimodules we are considering here are not, in general, central. As vector spaces they are "central" over the field $\kk$ only.\footnote{For central bimodules the multiplications by elements from the center $\mathfrak{Z}(R)$ coincide on both sides. Duals of central bimodules are again central.}
 \section{Cartan pairs as  noncommutative vector fields}
 We recall some definitions and properties announced in \cite{AB96,AB97}. We were focused on the algebraic construction of noncommutative analogue of vector fields originally called Cartan pairs.
 For this purpose, let us consider the vector space $\nd(R)$ as a bimodule with the following structure 
 $(r.A.r')(s)\doteq rA(s)r'$ \footnote{Some other possibilities will be discussed in the next section.}. It contains a subbimodule $\nd_0(R)\doteq \{A\in \nd(R)| A(1)=0\}$.
 \begin{definition}\label{rCp} 
 	A  \textit{right Cartan pair} over $R$ is a pair $(\X, \rtt)$ where:
 	\newline(1) $\X$ is an $R$--bimodule,
 	\newline(2)
 	$\rtt:\X\ni X \longrightarrow X^\rtt \in \nd (R)$   left module map, called an action 
 	\begin{equation}\label{rc1}
 		(r.X)^\rtt (s)=r X^\rtt (s)
 	\end{equation}   
 	(3) satisfying  the following relations
 	\begin{equation}\label{rc2}
 		X^\rtt (rs)= X^\rtt (r)s + (X.r)^\rtt (s)
 	\end{equation}
 	\newline(4) The map $\rtt$ is injective i.e., $\ker\rtt=0$.	 
 \end{definition}
 \begin{remark}\label{blad}  In \cite{AB96} the important condition (4) (cf. \cite{Llena03}) was optional and, unfortunately, not assumed in \cite{AB97}. One can observe that, in fact, $X^\rtt\in \nd_0(R)$, since $X^\rtt(1^2)=2X^\rtt(1)$ implies $X^\rtt(1)=0$. Thus to each right Cartan pair one can associate the short exact sequence of left modules
 	\begin{equation}\label{key3}
 		0\longrightarrow \X \stackrel{\rtt}{\longrightarrow} \nd_0(R)
 		\stackrel{\tilde\pi}{\longrightarrow}\tilde\X\longrightarrow0\,,
 	\end{equation}
 	which in general does not split, here $\tilde\X$ is a factor module 
 	$\nd_0(R)/\X$.	We will see in the next section that this sequence can be realized as a sequence of bimodules.
 \end{remark}
 A morphism of right Cartan pairs  $(\X_1, \rtt_1)$ and $(\X_2, \rtt_2)$ is a bimodule map $\Psi: \X_1\rightarrow \X_2$ such that $\rtt_2\circ\Psi=\rtt_1$.
 \begin{remark}\label{lCp} (Left Cartan pair)\ 
 	In similar manner one can define a left Cartan  pair $(\Y, \ltt)$ with an injective (right module action) map 
 	$\ltt:\Y\ni Y \longrightarrow Y^\ltt \in \nd_0 (R)$
 	\[
 	(Y.r)^\ltt (s)\ =\  Y^\ltt (s)r \]
 	and
 	\[Y^\ltt (r\, s)\ =r\,Y^\ltt (s) + (s.Y)^\ltt (r)\,, \]	  
 	or more naturally, in the inverse order, as acting on the left:  $(r\, s)^\ltt\!Y\ =r\, (s)^\ltt\!Y + (r)^\ltt\!(s.Y) $.
 \end{remark}
 \begin{theorem}\label{} To each FOC $(d, \Omega)$ one can associate
 	a well-defined right (respectively, left) Cartan pair $(\Omega^\trr,\p)$
 	(respectively, $(\Omega^\trl, )$), where $X^\p(r)=<X|dr>$ for any $X\in \Omega^\trr, r\in R$ (respectively, $Y^\eth(r)=<dr|Y>$ for any $Y\in \Omega^\trl, r\in R$).\\
 	If $\mathfrak{L} \in \hm_{(R,R)}(\Omega_1,\Omega_2)$ is a morphism  (i.e., $ d_2=\mathfrak{L}\circ d_1\,$) then its transpose $\mathfrak{L}^\trr \in \hm_{(R,R)}(\Omega^\trr_2,\Omega^\trr_1)$ is a morphism of right (respectively, left) Cartan pairs i.e.,  
 	$ \p_2=\p_1\circ\mathfrak{L}^\trr\,$. Moreover, if $\mathfrak{L}$ is an epimorphism then $\mathfrak{L}^\trr$ is a monomorphism.
 \end{theorem}
 \begin{proof}
 	We prove only a right-handed version. First we need to check properties from Definition \eqref{rCp}, e.g: \[(X.r)^\p (s)= <X.r| ds> = <X| r.ds> = 
 	<X| d(rs)- dr.s> =  X^\p(rs) -  X^\p(r) s \ .\]
 	Now, if $X\in\ker\p$ then $<X| dr>=0$ for each $r\in R$. But $dr$ generate $\Omega$ as a right module, so $X=0$.
 	If $ d_2=\mathfrak{L}\circ d_1\,$ then 
 	\[Y^{\p_2}(r)=<Y|d_2r>= <Y|\mathfrak{L}( d_1r)>=<\mathfrak{L}^\trr(Y)| d_1r>=\mathfrak{L}^\trr(Y)^{\p_1}(r)\,.\]
 	The last statement is a consequence of Proposition \ref{key1}.
 \end{proof}
 \begin{remark}
 	Our notation  $X^\p\equiv\p_X$ suggests the relation with  directional partial derivative along the vector field $X$. Indeed, it is the case at least for calculus on free (left or right) bimodules, cf. \cite{bk95a}. In this sense $(\Omega^\trr, \p)$ can be understood as a Cartan pair of right (directional) partial derivatives for FOC $(\Omega, d)$.\\
 \end{remark}
 \begin{example}\label{standard2} (Cartan pairs for commutative algebras)\ 
 	Assuming $R$ commutative i.e., $\mathfrak{Z}(R)=R$, and considering central bimodules with the left and right multiplications the same i.e.,  
 	$X.r=r.X$, we arrive to the  standard Leibniz rule 
 	\[X^\rtt (r\, s)\ =\ X^\rtt (r)\,s + r\, X^\rtt (s) \]
 	which is also satisfied in classical differential geometry of smooth manifolds. Therefore, $X^\rtt\in \mbox{Der}(R)$ becomes a derivation of the commutative algebra $R$. Moreover, the notions of left and  right Cartan pairs coincide in this case. Their duals  become central bimodules too. In fact, $\mbox{Der}(R)^\trr$  provides the module of differential one-forms for $R=C^\infty(M)$. 
 	In that sense, the definition of Cartan pairs generalizes the notion of vector fields on smooth manifolds (cf. Example \ref{standard}). More generally, it is known that $\mbox{Der}(R)\cong\Omega^1_K(R)^\trr$ for any commutative ring. Therefore, smooth differential one-forms can be considered as a bidual of the K\"ahler one-forms $(\Omega^1_K(R)^{\trr})^\trl$. Their dual gives back the derivation.
 \end{example}
 Any FOC $(d, \Omega)$ is related with the short exact sequence \eqref{univ1} of bimodules. Its right dual takes the form
 \begin{equation}\label{univ2}
 	0\longrightarrow \Omega^\trr\stackrel{\pi^\trr}{\longrightarrow}   
 	(\ker\mu)^\trr \longrightarrow  \bar\Omega^\trr\,, 
 \end{equation}
 where $\Omega^\trr$ stands for the bimodule corresponding to the right Cartan pair $(\Omega^\trr, \p)$.  Thus the second term $(\ker\mu)^\trr$ represents the universal right Cartan pair; 
 its structure will be described in the next section.  $\pi^\trr$ is a monomorphism. However,
 one should note that the last arrow is not an epimorphism in general (cf. Proposition \ref{key1}) what in the present context is not important.
 
 The natural question which appears now is if one can reconstruct FOC from an abstract (e.g. right) Cartan  pair? It appears to be not the case in general. Assume $(\X, \rtt)$ as in Definition \ref{rCp}. One can define a differential $d: R\rightarrow \Omega\equiv\X^\trl$ by the formula $<X|dr>\doteq X^\rtt(r)$ for each $X\in \X$. The Leibniz rule follows: $<X|d(rs)>\ =\ X^\rtt(rs)\ =\ X^\rtt(r)s\ +\ (X.r)^\rtt(s)\ =\ <X|dr>\!s+<X.r|ds>
 =<X|dr.s+r.ds>$. The only doubt is if elements $dr$ generate the bimodule
 $\X^\trl$. Indeed, dualizing \eqref{key3} one obtains a dual sequence of right modules
 \begin{equation}\label{k4}
 	0\longrightarrow \tilde\X^\trl \stackrel{\tilde\pi^\trl}{\longrightarrow} \nd_0(R)^\trl \stackrel{\rtt^\trl}{\longrightarrow}\X^\trl\,.
 \end{equation}
 According to general rules, the last arrow does not need to be an epimorphism. This shows that in a generic case it is not possible to reconstruct differential calculus from a Cartan pair.
 In other words, obtaining an epimorphism requires additional assumptions, e.g. that the left module $_R\! R$ is injective \cite{Kasch} or $\X$ is, finitely generated and projective \cite{Bourbaki}. The last corresponds to the reflexivity condition for which the formalism works fine.

 \section{Universal quantum vector fields}
 Our aim in this section is to apply duality functor to some exact sequences obtained  previously. We are focused on splitable universal sequences \eqref{f2}, \eqref{f4}. To this aim  we began with the following
 \begin{proposition}\label{u1} Right dual of a bifree bimodule $R\otimes R$ is
 	$(R\otimes R)^\trr=\nd\,(R)$ with a bimodule structure $(r.A*r')(s)\doteq rA(r's)$, for $A\in \nd\,(R)$. We shall denote this bimodule as $\nd\,(R)^\btr$ i.e., $(R\otimes R)^\trr=\nd\, (R)^\btr$\\
 	Similarly, $(R\otimes R)^\trl=\nd\,(R)^\btl$ with a bimodule structure $(r*A.r')(s)\doteq A(sr).r'$. 
 \end{proposition}
 \begin{proof}  We use the canonical isomorphism $\hm_{(-,R)}(R\otimes R, R)\equiv \nd\,(R)$ by $\tilde A(r\otimes r')=\tilde A(r\otimes 1)r'=A(r)r'$, where $\hm_{(-,R)}(R\otimes R, R)\ni\tilde A\leftrightarrow A\in \nd\,(R)$. Then a bimodule structure follows from the definitions \eqref{i3a}, \eqref{i9}.
 \end{proof}
 \begin{remark}\label{translations}
 	Multiplication from the left by elements $r\in R$ can be treated as an endomorphism $R\ni r\mapsto \L_r\in \nd(R)$ by
 	$\L_r(s)\doteq rs$ (a.k.a. left translations). In fact, this is a bimodule monomorphism since  $\L_r(1)=r$ and $\L_{r'rr"}=r'.\L_r*r"$. We denote its image as subbimodule $\L(_R\! R_R)\equiv R^\btr\subset \nd(R)^\btr$.	
 	Similarly, a multiplication from the right   $R\ni r\mapsto \R_r\in \nd(R)$ right translations by
 	$\R_r(s)\doteq sr$;   $\R_{r'rr''}=r'*\R_r.r''$. We denote its image as subbimodule $\R(_R\! R_R)\equiv R^\btl\subset \nd(R)^\btl$. They intersect in the center of $R$:
 	$\L(_R\! R_R)\cap \R(_R\! R_R)\cong \mathfrak{Z}(R)\supset\kk$.
 \end{remark}
 Summarizing, the vector space $\nd\,(R)$ admits a quatro-module structure with four mutually commuting multiplications (two from the left and two from the right):
 \newline(i) $r.A\doteq \L_r\circ A$,
 \newline(ii) $A*r\doteq A\circ\L_r $,
 \newline(iii) $r*A\doteq A\circ\R_r$,
 \newline(iv) $A.r\doteq \R_r\circ A$.\\
   They were used to construct bimodules $\nd\,(R)^\btr$ and $\nd\,(R)^\btl$. One observes that $\nd_0(R)$ can be embedded in $\nd\,(R)$ either as a left module with the multiplication (i) or as a right module with the multiplication (iv). Therefore, it is neither subbimodule of $\nd\,(R)^\btr$ nor of $\nd\,(R)^\btl$. Two pairs,  (i,iv) and (ii,iii), have central bimodule structures.
 
 The next elements in the sequences \eqref{f2}, \eqref{f4} are bimodules of universal differentials $_\chi(\bar R\otimes R)\cong (R\otimes \bar R)_\gamma \cong\ker\mu$. Their duals will represent universal Cartan pairs. Of course, right duals of isomorphic bimodules are isomorphic to each other by transpose isomorphism, respectively, left duals of isomorphic bimoduls are isomorphic as well.
 \begin{proposition}\label{u2}
 	(I.)\ $_\chi(\bar R\otimes R)^\trr=\hm (\bar R, R)\equiv \nd_0(R)$ with a bimodule structure $(r.A\odot r')(s)\doteq rA(r's)-rA(r')s$. Of course, these multiplications do preserve the condition $A(1)=0$. We shall denote this bimodule as $\nd_0^\rightthreetimes(R)$. It can be also considered as  a left submodule of $\nd(R)^\btr$.\\
 	(II.)\ $(R\otimes \bar R)_\gamma^\trl=\nd_0^\leftthreetimes (R)$ with a bimodule structure $(r\odot A.r')(s)\doteq A(sr)r'-sA(r)r'$.  Since the bimodules $_\chi(\bar R\otimes R)\cong (R\otimes \bar R)_\gamma$ are isomorphic then their right (respectively, left) duals are isomorphic too.
 \end{proposition}
 \begin{proof}  Using the splitting $R=\bar R\oplus \kk$ we have $\hm(R,R)=\hm(\bar R, R)\oplus \hm(\kk, R)$. $\hm(\kk, R)$ can be identify with $R$, while $\hm(\bar R, R)$ with $\nd_0(R)$. Indeed, if $A(1)=0$ then setting $A(\bar r) \doteq A(r)$ is well defined mapping. Then bimodule structures follow from the definitions  \eqref{i3a}, \eqref{i9}.
 \end{proof}
 \begin{remark}
 	In order to have a bimodule isomorphism $\mathfrak{S}:\nd_0^\ltt(R)\rightarrow \nd_0^\rtt (R)$, between left and right duals, one needs to find the element $\mathfrak{S}\in {\rm Iso}(\nd_0(R))$ such that for any $A\in \nd_0(R)$ and $r\in R$ the following two conditions:
 	\newline(i) \ \ $\mathfrak{S}(A\circ \R_r-\R_{A(r)})= \L_r\circ\mathfrak{S}(A)$\,,
 	\newline(ii) \ \ $\mathfrak{S}(\R_r\circ A)=\mathfrak{S}(A)\circ\L_r-\L_{\mathfrak{S}(A)(r)}$\,,
 	\newline are satisfied. However, in general, such mapping does not exist.
 \end{remark}
 It follows then that $\nd_0^\rightthreetimes(R)$  (respectively, $\nd_0^\leftthreetimes (R)$) can be treated as a bimodules of universal Cartan pairs (or universal noncommutative vector fields). Therefore, the actions $\rtt$  (respectively, $\ltt$) are, in fact, tautological actions. 
 Each right (respectively, left) Cartan pair is a subbimodule of the corresponding universal one, see \eqref{k4}. Finally, dualizing splitable short exact sequences \eqref{f2}, \eqref{f4} one obtains (cf. \eqref{univ2})
 \begin{proposition}\label{final} The following exact sequence of right modules
 	\begin{equation}\label{f2d}
 		0\longrightarrow \nd_0^\rtt (R)\stackrel{(D\otimes \id_R)^\trr}{\longrightarrow} 
 		\nd(R)^\btr
 		\stackrel{(1\otimes\id_R)^\trr}{\longrightarrow}R^\btr 
 		\longrightarrow 0
 	\end{equation}
 	(respectively, left modules)
 	\begin{equation}\label{f4d}
 		0\longrightarrow \nd_0^\ltt (R)\stackrel{(\id_R\otimes D )^\trl}{\longrightarrow} 
 		\nd(R)^\btl 
 		\stackrel{(\id_R\otimes 1)^\trl}{\longrightarrow}R^\btl 
 		\longrightarrow 0\,
 	\end{equation}
 	splits.
 \end{proposition}
 \begin{proof} 
 	Let $B\in\nd (R)$ be any endomorphism. It can be simply decomposed into
 	two parts as
 	\[B= (B-\L_{B(1)})\,+\, \L_{B(1)}\,\]
 	according to \eqref{f2d}. The first component of this splitting represent a projection into the right Cartan pair $\nd_0^\rtt (R)$. The  projection is a left module map. It is not a right module map since  multiplications from the right in $\nd_0^\rtt (R)$ and in $\nd(R)^\btr$ are different.  Similarly, 
 	\[B= (B-\R_{B(1)})\,+\, \R_{B(1)}\,.\]
 	represents the splitting \eqref{f4d}. 
 \end{proof}
 \begin{remark}\label{coord3} (Coordinate description, cf. \cite{bk95a}) 
 	Let us consider FOC $d\colon{R}\longrightarrow\Omega$, where $\Omega=_\beta\!(V\otimes R)$ is a right free module with left module structure determined by $\beta$. Utilizing the notation from Corollary \ref{coord1} one obtains $dr=\xi^i\otimes \p_i(r)$ with "right partial derivatives" $\p_i\in \nd_0 (R)$. Thus, due to the Lebniz rule, the following property holds $$\p_i(rs)=\p_i(r)s+\hat\beta_i^k(r)\p_k(s)\,.$$ One should note that for each $r\in R$ only the finite number of $\p_i(r)$ is non-vanishing. More precisely, taking into account \eqref{univ2}, one can write  $\p_i\in  \nd_0^\rtt (R)\cong (\ker\mu)^\trr$. Therefore, there exists some $\hat\p_i\in \Omega^\trr $ such that $\pi^\trr(\hat\p_i)=\p_i$. 
 	Indeed, there exist well defined functionals $\xi_i\in V^\bullet$ ("dual basis") such that $\xi_i(\xi^k)=\delta_i^k,  i,k\in I$ which are linearly independent \footnote{They provide a basis in $V^\bullet$ when $\dim V<\infty$. For $\dim V=\infty$ they  should be completed to a full basis in $V^\bullet$.}.
 	Next, we define elements $1\otimes \xi_i\in R\otimes V^\bullet$ 
 	\begin{equation}\label{coord4}
 		\hat \p_i\doteq 1\otimes \xi_i\in R\otimes V^\bullet \hookrightarrow (V\otimes R)^\trr\cong \Omega^\trr\hookrightarrow (\ker\mu)^\trr\cong\nd_0^\rtt (R)
 	\end{equation}	
 	due to \eqref{i5}, \eqref{iii5}, \eqref{univ2}. We recall that the papers \cite{bko94,bk95,bk95a,bk95b} study optimal calculi on coordinate algebras defined by generators and relations.
 \end{remark} 
 \begin{example}\label{ex1} Consider the internal derivation $d_1r=[\xi^1\otimes 1, r]=\xi^1\otimes r-\xi^k\otimes \hat\beta_k^1(r)$ in the bimodule from Corollary \ref{coord1}. Thus $\p_i=\delta_i^1\id_R-\hat\beta_i^1\in \nd_0(R)$. Considering the composition
 	$\p_i\circ\p_k=\delta_k^1\delta_i^1\id_R -\delta_i^1\hat\beta_k^1 - \delta_k^1\hat\beta_i^1 +\hat\beta_i^1\circ\hat\beta_k^1\in \nd_0(R)$ one finds that the second order differential operator with respect to the original  FOC becomes the first order for the universal one.
 \end{example}
 It is known that in the classical case linear differential operators of any finite order act on functions through their embedding into the endomorphisms algebra, see, e.g. \cite{Lunts97,GS09}. Particularly, they can be obtained by successive composition of vector fields.  
 Our results imply that the most general universal (linear) differential operators are of the first order.
 \begin{remark}\label{Der}
 	One should notice that for a commutative algebra $R$ not all elements of $\nd_0(R)$ are derivations. For example, if $Y,Y\in Der(R)\subset \nd_0(R)$
 	then $X\circ Y\in \nd_0(R)$ is not a derivation. 
 	In general, for noncommutative $R$,  the Lie algebra $Der(R)$ is a module over the center $\mathfrak{Z}(R)$ of $R$ \cite{Michor96}.
\end{remark}
 
 \section{ Summary}
 In this paper,  we realized again in more detail, using appropriate algebraic tools, the studies of noncommutative vector fields proposed originally in \cite{AB96,AB97} under the name of Cartan pairs. We emphasized the asymmetric behavior of this notion: to each FOC one can associate a bimodule of right/left vector fields. However, in contrast to the classical case, reconstruction of the initial FOC from such obtained left or right Cartan pair is not always  possible.
 We were focused on properties that are valid in the most general case.  A more detailed study should take into account
 specific properties of modules and/or the ring $R$. For example, projective (respectively, finitely generated) modules share most of the properties specific for vector spaces (respectively, finite-dimensional vector spaces). Also, so-called rings with perfect duality are easier to handle \cite{Kasch}. In general, they should be investigated on a case by case basis.  
 
 A modern approach to this problem takes into account the braided structure of monoidal categories related with some quasi-triangular Hopf algebras (quantum groups)  \cite{Aschieri20,Weber,Landi}  (see also earlier papers \cite{bmo92,Oz95,Lychagin}). These include also Woronowicz bicovariant calculi on Hopf modules (a.k.a Yetter-Drinfels modules) \cite{Woronowicz,Llena03,Peter96}.
 
 Higher-order quantum operators composed of quantum vector fields  and their relationships with deformations of Weyl algebra are not yet properly described, see, e.g. \cite{Lunts97,GMS05,GS09,Lychagin}. Also adopting the present framework to the category of Hilbert modules over (separable) $C^*$-algebras might be challenging, e.g. \cite{Manoh17}. Those are interesting subjects for future investigations.
 
 Finally, let us remark that dualizing, in a category of the corresponding vector spaces, the scheme presented above one can obtain codifferential calculi on coalgebras as well as some co-Cartan pairs called vector cofields. These can be done by reversing all arrows in diagramatic definitions of all objects and mappings involved. For example, reversing arrows in the definition of a unital associative algebra one obtains a counital coassociative coalgebra (a coring). Particularly, the notion of (left or right) codual for bicomodules over coalgebra can be introduced as well \cite{AB99}.
 
 \section*{Acknowledgments}
This research was inspired, in the middle of the 1990s, by numerous talks with Zbigniew, especially about Cartan type duality between one-forms and derivations on smooth manifolds, and by joint papers with Slava Kharchenko.

The current work has been supported by the Polish National Science Centre (NCN), project UMO-2017/27/B/ST2/01902. 
I would like to thank Thomas Weber for his interest and discussions during the Corfu  Workshop on Quantum Geometry, Field Theory and Gravity 2021.

\end{document}